\documentclass[journal,twoside,web]{ieeecolor}
\usepackage{lcsys}
\usepackage{cite,color} 
\usepackage{graphicx}
\usepackage{textcomp}
\usepackage{graphicx}   
\usepackage{amsmath,amssymb,amsfonts,xpatch,dsfont}
\newcommand{\mat}[1]{\left[\; \begin{matrix} #1 \end{matrix} \:\right]}

\newcommand{\simode}[1]{\left\{\;
\begin{aligned} #1 \end{aligned} \right.}
\newcommand{\bm}[1]{\boldsymbol{#1}}
\newcommand{\tr}{{\sf T}}

\newtheorem{theorem}{Theorem}

\newtheorem{remark}{Remark}
\newtheorem{lemma}{Lemma}

\def\BibTeX{{\rm B\kern-.05em{\sc i\kern-.025em b}\kern-.08em
    T\kern-.1667em\lower.7ex\hbox{E}\kern-.125emX}}
\begin{document}
\title{Structured Kalman Filter for Time Scale Generation in Atomic Clock Ensembles}
\author{Yuyue Yan, \IEEEmembership{Member, IEEE}, 
Takahiro Kawaguchi, \IEEEmembership{Member, IEEE}, 
Yuichiro Yano, \\ 
Yuko Hanado,   and Takayuki Ishizaki ,  \IEEEmembership{Member, IEEE}
\thanks{This work is supported by the Ministry of Internal Affairs and Communications (MIC) under its "Research and Development for Expansion of Radio Resources (JPJ000254)" program.  
Yuyue Yan and Takayuki Ishizaki are with the Department of Systems and Control Engineering, Tokyo Institute of Technology, Meguro, Tokyo 152-8552 Japan (e-mail:  yan.y.ac@m.titech.ac.jp, ishizaki@sc.e.titech.ac.jp). 
Takahiro Kawaguchi is with the Division of Electronics and Informatics, Gunma University, Kiryu, Gunma 371-8510 Japan.  
Yuichiro Yano and Yuko Hanado are with the National Institute of Information and Communications Technology,
Koganei, Tokyo 184-0015 Japan.}}

\maketitle

\begin{abstract}
In this article, we present a structured Kalman filter associated with the transformation matrix for observable Kalman canonical decomposition from conventional Kalman filter (CKF) in order to generate a more accurate time scale.
The conventional Kalman filter is a special case of the proposed structured Kalman filter which yields the same predicted unobservable or observable states when some conditions are satisfied.  
We consider an optimization problem respective to the transformation matrix where the objective function is associated with not only the expected value of prediction error but also its variance.
We reveal that such an objective function is a convex function and show some conditions under which CKF is nothing but the optimal algorithm if ideal computation is possible without computation error.
A numerical example is presented to show the robustness of the proposed method in terms of the initial error covariance. 
\end{abstract} 


\label{sec:introduction}
\IEEEPARstart{A}{n} atomic clock ensemble is a collection of highly accurate atomic clocks that work together to achieve a precise and stable timekeeping system. Atomic clocks are devices that measure time based on the vibrations of atoms with constant resonance frequencies. 
Using the measurements from several atomic clocks within an ensemble, the national metrology institutes (NMIs) all over the world can reduce these variations and create a more reliable and robust timekeeping system   \cite{galleani2010time,chan2009self,liu2021improving}.
The advancements in atomic clock technology and the development of accurate time scales based on clock ensembles offer numerous benefits and applications that are crucial for the future smart society, such as the satellite navigation \cite{wu2015uncertainty},  
financial networks with high-frequency trading and  
time-sensitive transactions \cite{mulvin2017media},
telecommunications \cite{seidelmann2011time}.

The variations in tick rates of atomic clocks are referred to as time deviations from the ideal clock behavior, which can be modeled as stochastic processes.
It has been revealed experimentally that the time deviations of the atomic clocks can be characterized as random noises in a stochastic linear differential equation \cite{zucca2005clock}.
Based on this, in the task of time generation, how to properly deal with the prediction problem for time deviations is the main issue to guarantee excellent performance of the generated time \cite{galleani2010time}. 
In order to improve the accuracy of atomic time and thus to benefit modern technology and fundamental physics, researchers in the time and frequency community have attempted to use Kalman filters for predicting those noises \cite{galleani2003use,galleani2010time,greenhall2012review,trainotti2022detection,mostafa2020enhancing}. 

The Kalman filter \cite{ma2018linear,xin2021kalman,9960730} is an optimal prediction algorithm minimizing an objective function associated solely with the mean-squared prediction error for detectable systems since in such a case the covariances of prediction errors are supposed to be bounded and converging to a steady-state value \cite{de1986riccati,shmaliy2017unbiased}. 
However, for an undetectable system, the covariance is growing unboundedly and thus both the expected value and covariance should be taken into account evaluating the performance of prediction, which implies that the conventional Kalman filter (CKF) may not be optimal.
In fact, as shown in \cite{yanPossibility}, the system model of the atomic clock ensembles is undetectable due to its intrinsic property that the phase difference is only measurable with fine resolution.
Even though some works have been done in terms of solving numerical instability  \cite{greenhall2012review,godel2017kalman}\footnote{Note that the method used in \cite{godel2017kalman} is the same as that presented by Greenhall in \cite{greenhall2012review}.} due to such undetectability,
it is worth asking whether there is space to improve the performance of the time scale from the conventional Kalman filter or not.

In this paper, different from the existing methods, we propose a structured Kalman filter associated with the transformation matrix for observable Kalman canonical decomposition of the atomic clock model.
The conditions where the proposed structured Kalman filter is reduced to the CKF in the unobservable or observable state space are presented.  
In such a case, the proposed method is understood as an alternative method of CKF algorithms avoiding numerical instability. 
By defining an objective function with respect to the transformation matrix according to the expected value and variance of \emph{atomic time} (i.e., prediction error in ensemble time deviation), an optimization problem is formulated to construct the optimal transformation matrix.
We reveal the fact that such an objective function is convex, and show some conditions under which the CKF is nothing but the optimal algorithm minimizing not only the square of expected value but also the variance of atomic time if ideal computation is possible without computation error.

\noindent \textbf{Notation}~We write $\mathbb{R}$ for the set of real numbers, $\mathbb{R}_+$ for the set of positive real numbers, 
$\mathbb{R}^{n}$ for the set of \emph{n}$\times 1$  real column vectors,  and $\mathbb{R}$$^{n \times m}$ for the set of \emph{n}$\times m$ real matrices.
Moreover, $\otimes$ denotes the Kronecker product, $(\cdot)^\tr$ denotes transpose,   $(\cdot)^\dagger$ denotes the Moore-Penrose inverse,   and
$ \mathsf{diag}(\cdot)$ denotes a diagonal matrix. 
Finally, $ \mathbb{E}[x]$ denotes the expected value of a random variable $x$,  whereas  $\mathds{1}_n $ and $I_{n}$  denote the ones (column) vector and the identity matrix of dimension $n$, respectively.

\section{Preliminaries and Problem}\label{sec:moddes}
Consider the $n$th order linear stochastic discrete-time system for an $m$-clock homogeneous ensemble given by
\begin{equation}\label{eq:Ndmodel}
\Sigma:    \simode{
    \bm{x}[k+1] &=  \bm F\bm{x}[k] + \bm{v}[k]  \\
    \bm{y}[k] &= \bm H \bm{x}[k] + \bm{w}[k] \\
    \bm{z}[k] &=\bm D\bm{x}[k]
    }
\end{equation}   
where 
$\bm{x}[k]:= (\bm{x}_1^\tr[k],  \ldots, \bm{x}_n^\tr[k])^\tr\in\mathbb R^{nm}$,
$\bm{y}[k]\in\mathbb R^{m-1}$, 
are the state and measurement, respectively, with $\bm{x}_i[k]:=(x^1_i[k] ,\ldots, x^m_i[k])^\tr\in\mathbb R^{m}$; 
$\bm F:= A\otimes I_m\in\mathbb R^{nm\times nm}$, 
$\bm H:= C \otimes \overline{V}\in\mathbb R^{(m-1)\times nm}$
are the system state transition matrix and observation matrix defined as
\begin{align} 
A:=&A_\tau:=\mat{
    1 & \tau & \frac{\tau^2}{2} & \cdots & \tfrac{\tau^{n-1}}{(n-1)!} \\
    0 & 1 & \tau & \cdots & \frac{\tau^{n-2}}{(n-2)!} \\
    \vdots & & \ddots & \ddots & \vdots \\
    \vdots & & & 1 & \tau \\
    0 & 0 & \cdots & \cdots & 1
    }  
\\ 
C :=& \mat{1 & 0 & \cdots & 0}\in\mathbb R^{1\times n}
\\ 
\overline{V}:=&\mat{I_{m-1}& -\mathds{1}_{m-1}}\in\mathbb R^{(m-1)\times m}
\end{align}  
for the sampling interval $\tau$;
$\bm{v}[k]\in\mathbb R^{nm}$, $\bm{w}[k]\in\mathbb R^{m-1}$ are the system noise and observation noise, and they are white Gaussian with the covariances \cite{galleani2010time} being
\[
W := \int_{0}^{\tau}A_t\mathsf{diag}(q_1^2,\ldots,q_n^2)A_t^{\sf T}dt \otimes I_m, \quad
R:=r^2I_{m-1}
\]
respectively.   
In such a homogeneous ensemble, each atomic clock works as an independent oscillator with the number of waves counted as the clock reading of the clocks, where there exists a slight difference (referred to as the time deviation) between  the actual clock reading and the ideal clock reading.

Note that in the system model (\ref{eq:Ndmodel}),  the vector $\bm{x}_1[k]$ represents the time deviations of the clocks \cite{yanrelation23} and hence the scalar $\bm{z}[k]\in\mathbb R$ denotes the ensemble time deviation defined as a weighted value from the individual time deviations with 
  \begin{align} 
\bm D:= \tfrac{1}{m}C\bigl( I_n \otimes\mathds{1}_m^{\tr})
\end{align}
where the weights are considered to be all equal. 
Furthermore, $Q$ is the covariance of the individual system noise where $A_t$ in the integrant is understood as $A_\tau$ with $\tau$ replaced by $t$.
In terms of the physical meanings, the measurement $\bm{y}[k]$ is the difference of the time deviations (or, equivalently, the difference of the clock reading) measured independently between the clocks where clock $m$ plays the role as a reference clock
(see the detailed model in \cite{galleani2010time,yanrelation23}).

Even though the detectability condition is not satisfied in the system model \cite{yanPossibility}, CKF algorithms are still used in  many practice scenarios of time scale generation of clock ensembles \cite{greenhall2012review,galleani2010time,trainotti2022detection}. 
The CKF algorithm is
\begin{align}  \label{eq:Hanado1}
  & \bm L_k = -\bm F\bm P_k {{\bm H} ^{\sf T}\bigr(\bm{H}\bm P_k{ \bm H}^{\sf T}+R\bigr)^{-1}}  \\  \label{eq:Pk}
  &\bm P_{k+1}=(\bm F +\bm L_k {\bm H})\bm P_k\bm F^{\sf T}+W  \\
  \label{eq:Hanado2}
  &\hat{\bm x}[k+1] =\bm F \hat{\bm x}[k]  - \bm L_k(\bm{y}[k]-\bm{H}\hat{\bm x}[k])
\end{align}
for $k\geq 0$ with the initial $\bm P_0=pI_{nm}$ for some constant $p\in\mathbb{R}_+$, where $\bm L_k$ and $\bm P_k$ are the Kalman (observer) gain and error covariance. 
Together with the linear-quadratic regulator, the Kalman filter solves the linear–quadratic–Gaussian control problem minimizing the mean squared error 
$
\mathbb E\big[\|\bm x[k]-\hat{\bm x}[k]\|^2\big]
$.
Using the Kalman filter, the ensemble time deviation ${\bm z}[k]$ is predicted as $\hat{\bm z}[k]:=\bm D \hat{\bm x}[k]$ and hence generated clock reading of the $m$-clock ensemble is given as
\begin{align}      
\hat h_0[k]=\tfrac{1}{m}\sum\nolimits_{i=1}^m h^j[k]-\hat{\bm z}[k]\in\mathbb R
\end{align}  
where $h^j[k]$ is the actual clock reading of clock $j$ and hence the ideal clock reading can be written as $h_0[k]=\hat h_0[k]\in\mathbb R$ with $\hat{\bm z}[k]$ replaced by ${\bm z}[k]$.
The accuracy of the generated clock reading is evaluated by \emph{atomic time} (which is referred to as \emph{Temps Atomique} (TA) in Europe \cite{galleani2010time}) as 
\begin{equation}\label{eq:TA} 
{\rm TA}[k] := h_0[k]-\hat h_0[k]  =\bm z[k]-\hat{\bm z}[k]= \bm D(\bm{x}[k] -\hat{\bm{x}}[k])
\end{equation}
which is depending on how we yield the predicted state $\hat{\bm x}[k]$.

\emph{Problem}:
Consider an $m$-clock ensemble to generate an accurate time scale for a given initial guess $\hat{\bm x}[0]$.
A fundamental question is whether there is a space to further improve the performance of the time scale from the CKF algorithms, i.e., further
decrease  the atomic time 
${\rm TA}[k] $ as much as possible.

\section{Proposed Structured Kalman Filter}
This section introduces a structured Kalman filter motivated by observable Kalman canonical decomposition for the atomic clock model.       
Specifically, defining a transformation matrix $\varGamma\in\mathbb R^{n\times n(m-1)}$, we consider
\begin{equation} \label{eq:tranformation}
\bm x=\underbrace{\mat{I_n\otimes \overline{V}^{\dagger}&I_n\otimes\mathds{1}_m}\mat{ I_{n(m-1)}& 0\\ \varGamma&I_{n}}}_{
\mathbb T}
\underbrace{\mat{ \bm \xi_{\rm o}   \\  \bm{\xi}_{\rm \bar o}  }}_{ \bm \xi}
\end{equation}
to make decomposition of the system $\Sigma$ into the observable state $\bm \xi_{\rm o}$ and unobservable state $\bm{\xi}_{\rm \bar o}$ for constructing a Kalman filter and a predictor, respectively. 
With such a structure of $\mathbb T$, the basis selection of observable states is limited without $\varGamma$ since all possible basis selections can be represented by introducing $\varGamma$.
Our main idea of including the matrix $\varGamma$ in the linear transformation \eqref{eq:tranformation} is to keep the variability of the basis selections so that one can discuss an optimization problem introduced later for improving the performance of the generated atomic time ${\rm TA}[k]$.
Then, we have
\begin{equation} 
\tilde\Sigma:\simode{
\bm \xi [k+1] &=\mat{\bm F_{\rm oo}   &  0 \\\bm F_{\rm \bar o o}  &  \bm F_{\rm \bar o  \bar o} } \bm \xi[k]+ \mathbb T^{-1}\bm{v}[k] \\
\bm{y}[k] &=\bigl[\bm H_{\rm o}   \ \ \   0\bigl]   \bm \xi [k] + \bm{w}[k] \\
\bm{z}[k] &=\bm D\mathbb T\bm{\xi}[k]}
\end{equation}
where $\bm H_{\rm o}:=C \otimes I_{m-1}$, 
\begin{equation} \label{system_equations}
\simode{ \bm F_{\rm oo} &:= A\otimes I_{m-1}, \\
\bm F_{\rm\bar o\bar o}&:=A,\\
\bm F_{\rm\bar oo}&:=-\varGamma(A\otimes I_{m-1})+A\varGamma ,
}
\end{equation}
and the covariance of the state noise affecting the observable state $\bm \xi_{\rm o}$ is given by  
$
W_{\rm o}:=(I_n\otimes\overline{V})W(I_n\otimes  \overline{V})^{\sf T}=Q\otimes\overline{V}\overline{V}^{\sf T}.
$

Now we introduce our structured Kalman filter as follows.
\subsubsection{Initialization} 
A guess of the initial state $\hat {\bm x}[0]$ is transformed into $\hat{\bm \xi} [0] = T^{-1} \hat {\bm x}[0]$.

\subsubsection{Update for observable states}
Construct
\begin{align} \label{eq:Kalmangain}
&\hat{\bm L}_k = -\bm F_{\rm oo}\hat{\bm P}_k {\bm{H}_{\rm o} ^{\sf T}\bigr(\bm{H}_{\rm o}\hat{\bm P}_k\bm{ H}_{\rm o}^{\sf T}+R\bigr)^{-1}}  \\   \label{eq:Lk_PKF}
&\hat{\bm P}_{k+1}=(\bm F_{\rm oo} +\hat{\bm L}_k \bm{H}_{\rm o})\hat{\bm P}_k\bm F_{\rm oo} ^{\sf T} +W_{\rm o}  
\\ \label{eq:PKA2_o}
&\hat{\bm{\xi}}_{\rm o}[k+1] =\bm F_{\rm oo}  \hat{\bm{\xi}}_{\rm o}[k]-\hat{\bm L}_k(\bm{y}[k]-\bm{H}_{\rm o} \hat{\bm{\xi}}_{\rm o}[k]) 
\end{align}
for $k\geq 0$ with some initial condition $\hat{\bm P}_0$, where $\hat{\bm L}_k$ and $\hat{\bm P}_k$ are the Kalman gain and error covariance for observable state.

\subsubsection{Prediction for unobservable states} 
Construct
\begin{align}  \label{eq:PKA2_uo}
\hat{\bm{\xi}}_{\rm \bar o}[k+1]
= 
  \bm F_{\rm \bar oo} \hat{\bm{\xi}}_{\rm o}[k] 
+\bm F_{\rm \bar o\bar o} \hat{\bm{\xi}}_{\rm \bar o}[k].
\end{align}

\subsubsection{Update in original coordinate} Update the predicted state  according to the linear transformation    
\eqref{eq:tranformation}, i.e., 
\begin{equation}  \label{eq:PKA1}
\hat {\bm x}[k]=\mathbb T\hat{\bm \xi}[k].
\end{equation}

\section{Main Results}\label{sec:CKF} 

Before we present the main results for the proposed structured Kalman filter, we show a lemma for CKF algorithms.

\begin{lemma}\label{lem:2}
If the initial error covariance $\bm P_0$ satisfies $\bm P_0(I_n\otimes  \overline{V})^\dagger=(I_n\otimes  \overline{V})^\dagger\check{P}_0$
for some $\check{P}_0$,
then
the error covariance of CKF satisfies
$\bm P_k(I_n\otimes  \overline{V})^\dagger=(I_n\otimes  \overline{V})^\dagger\check{\bm P}_k$~with 
\begin{align}\label{lemma1:1}
\check{\bm P}_{k+1}=&(\bm F_{\rm oo}
                          +\check{\bm G}_k\bm H_{\rm o})\check{\bm P}_k\bm F_{\rm oo}^{\sf T}
                         +Q\otimes I_{m-1}
\\   \label{lemma1:2}
\check{\bm G}_k=&-\bm F_{\rm oo}\check{\bm P}_k\bm H_{\rm o}^{\sf T}\bigr(\bm H_{\rm o}\check{\bm P}_k\bm H_{\rm o}^{\sf T}
                          +r(\overline{V}\overline{V}^{\sf T})^{-1}\bigr)^{-1}
\end{align} 
for $\check{\bm P}_0=pI_{nm-n}$ and any $k=0,1,2,\ldots$.
\end{lemma}

\begin{proof}
The proof is shown in \cite{yanstructured23}.
\end{proof} 

\begin{theorem}   \label{prop:partial_kalman_obse}
Consider an $m$-clock ensemble to generate an accurate time scale for a given initial guess $\hat{\bm x}[0]$.
If $\varGamma=0$ (resp., $\hat{\bm P}_0=(I_n\otimes\overline{V})\bm P_0(I_n\otimes\overline{V})^{\tr}$), 
then the structured Kalman filter \eqref{eq:Kalmangain}--\eqref{eq:PKA1} ideally yield the same predicted unobservable state $\hat{\bm{\xi}}_{\rm \bar o}$ (resp., observable state $\hat{\bm{\xi}}_{\rm o}$) as CKF.
In particular, if $\varGamma=0$ and $\hat{\bm P}_0=(I_n\otimes\overline{V})\bm P_0(I_n\otimes\overline{V})^{\tr}$, 
then the structured Kalman filter ideally reduces the same as CKF algorithms yielding the same predicted state $\hat{\bm x}[k]$ for all $k \in \{0,1,2,\ldots\}$.
\end{theorem} 

\begin{proof} 
First, we prove equivalency of the unobservable states $\hat{\bm{\xi}}_{\rm \bar o}[k]$ in the two methods. 
Note from $\mathds{1}_{m}^{\sf T}\overline{V}^\dagger=0$ and Lemma~\ref{lem:2}
that
\begin{align}\nonumber
0=&(A\!\otimes\!\mathds{1}_m^{\tr})(I_n\!\otimes\!  \overline{V})^\dagger\check{\bm P}_k\bm H_{\rm o}(I_n\!\otimes\!\overline{V}\overline{V}^{\sf T}\!)\bigr(\bm{H}\bm P_k{ \bm H}^{\sf T}\!\!+\!R\bigr)^{-1}
\\ \nonumber
=&(A\!\otimes\!\mathds{1}_m^{\tr}) \bm P_k(I_n\!\otimes\!  \overline{V})^\dagger\bm H_{\rm o}(I_n\!\otimes\!\overline{V}\overline{V}^{\sf T}\!)\bigr(\bm{H}\bm P_k{ \bm H}^{\sf T}\!\!+\!R\bigr)^{-1}\\ \nonumber
=&(I_n\otimes\mathds{1}_m^{\tr})\bm F \bm P_k {\bm H} ^{\sf T}\bigr(\bm{H}\bm P_k{ \bm H}^{\sf T}+R\bigr)^{-1}
\\  \label{eq:zero_equality}
=&-(I_n \otimes\mathds{1}_m^{\tr})\bm L_k,\quad k=0,1,2,\ldots
\end{align}
Then, since the prediction error $\bm{\epsilon}:=\bm x-\hat{\bm x}$ of CKF is   
\[
\bm{\epsilon} [k+1]=(\bm F+\bm L_k\bm{H})\bm{\epsilon}[k]+\bm L_k\bm w[k]+\bm v[k] 
\]
it follows that the prediction error $\bm{\epsilon}_{\rm \bar o}:={\bm{\xi}}_{\rm \bar o}-\hat{\bm{\xi}}_{\rm \bar o}=\bigl(I_n \otimes\mathds{1}_m^{\tr})\bm\epsilon$ of unobservable state in CKF is given by 
\begin{align}\nonumber
\bm{\epsilon}_{\rm \bar o}[k+1] 
=&A\bm{\epsilon}_{\rm \bar o}[k]+\bigl( I_n \otimes\mathds{1}_m^{\tr}\bigr) \bm{v}[k]   
\\  \nonumber
&+(I_n\otimes\mathds{1}_m^{\tr})\bm L_k(\bm{H} \bm{\epsilon}[k]+ \bm w[k])
\\ \label{eq:weighted_error}
=&A\bm{\epsilon}_{\rm \bar o}[k]+\bigl( I_n \otimes\mathds{1}_m^{\tr}\bigr) \bm{v}[k].   
\end{align} 
For the structured Kalman filter,
the predicted state $\hat{\bm{x}}$ follows
\begin{align}\nonumber
\hat{\bm{x}}[k+1] = & \mathbb T\left( \mat{\bm F_{\rm oo} +\hat{\bm L}_k\bm H_{\rm o}  \!\!  &  0 \\ \bm F_{\rm \bar o o} \!\! & \bm F_{\rm \bar o \bar o} }\right)\mathbb T^{-1}    \hat{\bm{x}}[k]- \bm G_k\bm{y}[k]
\\ \nonumber  
=&\Big( \bm F+\mathbb T \mat{\hat{\bm L}_k\bm H_{\rm o}   &  0 \\ 0  & 0} \mathbb T^{-1} \Big)   \hat{\bm{x}}[k]- \bm G_k\bm{y}[k]\\  \nonumber
=& ( \bm F+\bm G_k \bm H )   \hat{\bm{x}}[k]\!-\!\bm G_k(\bm H{\bm{x}}[k]+\bm{w}[k])
\end{align}
where the matrix $\bm G_k$ is defined as
\[
\bm G_k:=\mathbb T\mat{\hat{\bm L}_k\\ 0 }=\Big[I_n\otimes \overline{V}^{\dagger}+(I_n\otimes  \mathds{1}_m)\varGamma\Big]\hat{\bm L}_k.
\]
Then it follows from the prediction error given by 
\begin{align}\label{eq:erro_dynamics_Kalman2}
{\bm{\epsilon}}[k+1] =&(\bm F+\bm G_k \bm H){\bm{\epsilon}}[k]+\bm G_k  \bm{w}[k]+ \bm{v}[k]
\end{align}
that the prediction error $\bm{\epsilon}_{\rm \bar o}$ of unobservable state in structured Kalman filter is given by
\begin{align}  \nonumber
\bm{\epsilon}_{\rm \bar o}[k+1]=&A\bm{\epsilon}_{\rm \bar o}[k]+\bigl(I_n\otimes\mathds{1}_m^{\tr}\bigr) \bm{v}[k] 
\\ \label{eq:erro_dynamics_Kalman}
&+m\varGamma\hat{\bm L}_k(\bm{H}\bm{\epsilon}[k]+\bm w[k]) 
\end{align}
which reduces the same as \eqref{eq:weighted_error} for CKF when $\varGamma=0$.

Next we prove equivalency of observable states $\hat{\bm{\xi}}_{\rm o}[k]$ in the two methods. 
Note that error covariance of the observable state ${\bm{\xi}}_{\rm o}[k]$ of CKF is
\begin{align}\nonumber
\bm P_{\rm o}[k+1] =&(I_n\otimes \overline{V})\bm P_{k+1}(I_n\otimes \overline{V})^{\sf T}
                                 =\check{\bm P}_{k+1}(I_n\!\otimes\!\overline{V}\overline{V}^{\sf T})
                                 \\ \nonumber
                              =&(\bm F_{\rm oo}
                          +\check{\bm G}_k\bm H_{\rm o})\check{\bm P}_k\bm F_{\rm oo}^{\sf T}
                                (I_n\otimes\overline{V}\overline{V}^{\sf T})
                                +W_{\rm o}
\\ \nonumber
                             =&(\bm F_{\rm oo}
                          +\check{\bm G}_k\bm H_{\rm o})\check{\bm P}_k
                                 (I_n\otimes\overline{V}\overline{V}^{\sf T})
                                 \bm F_{\rm oo}^{\sf T}
                                +W_{\rm o}
                                \\ \label{eq:Lk_CKF_o}
                             =&(\bm F_{\rm oo}
                          +\check{\bm G}_k\bm H_{\rm o}){\bm P}_{\rm o}[k]
                                 \bm F_{\rm oo}^{\sf T}
                                +W_{\rm o}
\\ \nonumber                         
\check{\bm G}_k=&-\bm F_{\rm oo}\bm P_{\rm o}[k]\bm H_{\rm o}^{\sf T}(\overline{V}\overline{V}^{\sf T})^{-1}
                               \\  \nonumber
                             &\cdot\bigr(\bm H_{\rm o}\bm P_{\rm o}[k]\bm H_{\rm o}^{\sf T}
                             (\overline{V}\overline{V}^{\sf T})^{-1}
                             +r(\overline{V}\overline{V}^{\sf T})^{-1}\bigr)^{-1}
                            \\ \label{eq:Kalmangain_CKF_o}
                             =&-\bm F_{\rm oo}\bm P_{\rm o}[k]\bm H_{\rm o}^{\sf T}
                             \bigr(\bm H_{\rm o}\bm P_{\rm o}[k]\bm H_{\rm o}^{\sf T}\!+\!R\bigr)^{-1}
\end{align}
which indicate along with condition
$\hat{\bm P}_0=(I_n\otimes\overline{V})\bm P_0( I_n\otimes\overline{V})^{\tr}$  
that $\bm P_{\rm o}[k+1]=\hat{\bm P}_{k+1}$ for the error covariance $\hat{\bm P}_{k+1}$ defined in \eqref{eq:Lk_PKF} and $\check{\bm G}_k=\hat{\bm L}_k$ for the Kalam gain $\hat{\bm K}_k$ defined in \eqref{eq:Kalmangain}.
Thus, the prediction error $\bm{\epsilon}_{\rm o}:={\bm{\xi}}_{\rm o}-\hat{\bm{\xi}}_{\rm o}=\bigl(I_n \otimes\overline{V})\bm\epsilon$ of observable state under CKF is the same as structured Kalman filter, which completes the proof. 
\end{proof}

\begin{remark}  
It follows from the definition of atomic time that 
${\rm TA}[k]=\frac{1}{m}C\bm{\epsilon}_{\rm \bar o}[k]$.
Thus, Theorem~\ref{prop:partial_kalman_obse} implies that ${\rm TA}[k]$ under CKF and structured Kalman filter are the same when the condition $\varGamma=0$ is satisfied, 
and they are ideally independent from the measurement noises $\bm{w}[0],\ldots,\bm{w}[k]$ (see the expression of $\bm{\epsilon}_{\rm \bar o}$ in \eqref{eq:weighted_error}).
In addition, it is important to note that both of error covariance \eqref{eq:Lk_CKF_o} along with \eqref{eq:Kalmangain_CKF_o} and error covariance \eqref{eq:Lk_PKF} along with \eqref{eq:Kalmangain} are converging to the same steady state. 
Therefore, even for the case of 
$ 
\hat{\bm P}_0\not=(I_n\otimes \overline{V})\bm P_0(I_n\otimes\overline{V})^{\tr}$,
the updated observable states $\hat{\bm{\xi}}_{\rm o}$ of those two methods are eventually the same for $\varGamma$ if there is no computation error.
In other words, the essential difference between the CKF and our structured Kalman filter algorithms comes from the third term of the the unobservable prediction error $\bm{\epsilon}_{\rm\bar o}$ in \eqref{eq:erro_dynamics_Kalman}.
\end{remark}   
   
\begin{remark}    
The CKF algorithms may result in numerical instability problems in practical implementation for the undetectable system of the atomic clock ensemble.
This is because the entries of the error covariance matrix and the computational errors of the CKF algorithms grow  unboundedly \cite{greenhall2012review,godel2017kalman}. 
However, Theorem~\ref{prop:partial_kalman_obse} indicates that the proposed structured
Kalman filter can be used as an alternative method of CKF algorithms avoiding 
numerical instability. 
\end{remark}   

Recalling from \eqref{system_equations} and \eqref{eq:PKA2_uo} that the proposed method is associated with the transformation matrix $\varGamma$, which affects the predictions of unobservable states within the transformed system and hence influences the performance of the atomic time. 
Therefore, although the proposed
structured Kalman filter with $\varGamma=0$ is an alternative to the CKF algorithms, it may be possible to further improve the time scale by addressing an optimization problem with respect to the transformation matrix.
In such a case, the conditions for determining the optimal $\varGamma$ are derived based on the prediction error and the error covariance resulting from the application of Kalman filter to the transformed system.  
 
\begin{theorem}    \label{thm:d4}  
Consider an $m$-clock ensemble with the structured Kalman filter algorithm \eqref{eq:Kalmangain}--\eqref{eq:PKA1} and the measurement 
$\{\bm{y}[k]: k=0,1,\ldots, T\}$.
It follows that the cost function
\begin{equation}\label{eq:payoff}
J(\varGamma):=\sum\nolimits_{k=0}^T\Big(\delta_1(\mathbb E[{\rm TA}[k]])^2+\delta_2\mathbb V[{{\rm TA}}[k]]\Big) 
\end{equation}
is a convex function with respect to $\varGamma$ for any $\delta_1,\delta_2\geq0$, 
$\delta_1+\delta_2\not=0$,
where $\mathbb V[{{\rm TA}}[k]]$ represents the variance of atomic time ${\rm TA}[k]$ for a given $\varGamma$.
Furthermore,  if the initial predicted state $\bm{\epsilon}[0]\sim\mathcal N(\mu_0,Q_0)$ satisfies  
$Q_0=\hat Q_0\otimes pI_m$ and 
$\mu_0=\hat{ \mu}_0\otimes\mathds{1}_m$ for some positive definite matrix $\hat Q_0=\hat Q_0^{\tr}\in\mathbb R^{n\times n}$, 
some $p\geq0$, 
and some vector $\hat{ \mu}_0\in\mathbb R^{n}$, 
then $\varGamma=0$ minimizes the cost function \eqref{eq:payoff}
in structured Kalman filter. 
      
\end{theorem}   
\begin{proof}
First of all, let $\lambda_0^k,\ldots,\lambda_k^k$ be defined as $\lambda_k^k=I$ and $\lambda_i^k:=\prod\nolimits_{j=k}^{i+1} (\bm F + \bm G_j \bm H)=\prod\nolimits_{j=k}^{i+1}[\bm F +(\bm G_j^1+\bm G_j^2)\bm H]$ for $ i<k$, 
where $\bm{G}_j^1:=(I_n\otimes \overline{V}^{\dagger}) \hat{\bm L}_j$ and $\bm{G}_j^2:= (I_n\otimes\mathds{1}_m)\varGamma\hat{\bm L}_j$.
Using $\bigl(I_n\otimes\mathds{1}_m^{\tr}\bigr)\bm{G}_k^1=0$, 
$\bigl(I_n\otimes\mathds{1}_m^{\tr}\bigr)\bm{G}_k^2=\varGamma\hat{\bm L}_k$, 
$\bm H\lambda_i^k=\bm H\prod\nolimits_{j=k}^{i+1}(\bm F+\bm{G}_j^1\bm H)$, $i<k$,
it is obtained from mathematical induction \cite{yanstructured23} that 
\begin{align} \nonumber
&\bm D\lambda_i^k
=\tfrac{C}{m}\bigl(I_n\otimes\mathds{1}_m^{\tr}\bigr)\lambda_i^k
=\tfrac{C}{m}\bigl[A\bigl(I_n\!\otimes\!\mathds{1}_m^{\tr}\bigr)+m\varGamma\hat{\bm K}_k\bm H\bigr]\lambda_i^{k-1} 
\\ \nonumber
&=C\sum\nolimits_{\delta=0}^{k-i-1}\Big( A^\delta\varGamma \hat{\bm K}_{k-\delta } \underbrace{\bm H  \prod\nolimits_{j=k-\delta-1}^{i+1}(\bm F+
\bm{G}_j^1\bm H)}_{\bm \Phi^i_{k-\delta}}\Big) 
\\   \label{eq:proof_KSk} 
&\quad +\tfrac{C}{m}A^{k-i}\bigl(I_n\otimes\mathds{1}_m^{\tr}\bigr),\quad i<k,
\\ \nonumber
&\bm D\lambda_i^k\bm G_i
=C\sum\nolimits_{\delta =0}^{k-i-1}\!\!A^\delta\varGamma \hat{\bm K}_{k-\delta }\bm \Phi^i_{k-\delta}\bm G_i^1
\\
\label{eq:fffs2} 
&\quad\quad\quad\quad\ +\tfrac{C}{m}A^{k-i}\bigl(I_n\otimes\mathds{1}_m^{\tr}\bigr)\varGamma \hat{\bm K}_i
,\quad  i< k.
\end{align}
which are linear to $\varGamma$.

Note from \eqref{eq:erro_dynamics_Kalman2} that the expected value and variance of the prediction error are 
$\mathbb E[\bm{\epsilon}[k]]=\lambda_0^k \mu_0$, 
$\mathbb V[{\bm{\epsilon}} [k]]=\lambda_0^kQ_0(\lambda_0^k)^{\tr} +\sum\nolimits_{i=1}^k\lambda_i^k(\bm G_iR\bm G_i^{\tr} +W)(\lambda_i^k)^{\tr}$, respectively.
Thus, \eqref{eq:proof_KSk}, \eqref{eq:fffs2} imply that the expected value 
$\mathbb E[{\rm TA}[k]]= \bm D \lambda_0^k  \mu_0$ is linear to $\varGamma$, whereas the variance 
$\mathbb V[{{\rm TA}}[k]]= \bm D\mathbb V[{{\bm{\epsilon}}}[k]]\bm D^{\tr}=\bm D\lambda_0^kQ_0(\lambda_0^k)^{\tr}\bm D^{\tr} +\sum\nolimits_{i=1}^k\bm D\lambda_i^k(\bm G_iR\bm G_i^{\tr} +W)(\bm D\lambda_i^k)^{\tr}$
is quadratic to $\varGamma$.
Now, since the variance $\mathbb V[{{\rm TA}}[k]]$ is nonnegative, it follows that the function $\mathbb V[{{\rm TA}}[k]]$ is convex to $\varGamma$ 
and hence the cost function $J(\varGamma)$ is a convex quadratic function in $\varGamma$ for any $\delta_1,\delta_2\geq0$, $\delta_1+\delta_2\not=0$. 
Then, due to the convexity and differentiability of the cost function \eqref{eq:payoff} with $\varGamma\in\mathbb R^{n\times m(n-1)}$ in unbounded space, 
the matrix $\varGamma=0$ is optimal if and only if there is no linear term in neither $\mathbb E[{\rm TA}[k]]$ nor $\mathbb V[{{\rm TA}}[k]]$ for $\varGamma$. 
Hence, the result is immediate since it follows from \eqref{eq:proof_KSk}, \eqref{eq:fffs2}  that possible linear terms in $\mathbb E[{\rm TA}[k]]$ and $\mathbb V[{{\rm TA}}[k]]$ given by 
\begin{align} 
   \nonumber 
&\sum\nolimits_{\delta =0}^{k-1} C A^\delta\varGamma \hat{\bm K}_{k-\delta } 
\underbrace{\bm \Phi^0_{k-\delta}\mu_0}_{=0} 
\\     \nonumber 
&\tfrac{1}{m}\sum\nolimits_{\delta =0}^{k-1}C  A^\delta\varGamma \hat{\bm K}_{k-\delta }   
\underbrace{\bm \Phi^0_{k-\delta}  Q_0\bigl( I_n \otimes\mathds{1}_m \bigr)}_{=0}{(CA^k)}^{\tr} 
\end{align}
are all zero when $Q_0=\hat Q_0\otimes pI_m$ and $\mu_0=\hat{ \mu}_0\otimes\mathds{1}_m$ because
\[
\bm \Phi^i_{k-\delta}
= \sum\nolimits_{j=1}^{2^{k-\delta-i-1}}\!\! \bm \phi_j \bm H \bm F^{\lambda_j}
= \sum\nolimits_{j=1}^{2^{k-\delta-i-1}}\!\! \bm \phi_j (CA^{\lambda_j} \otimes \overline{V})
\]
for some matrices $\bm \phi_j$ and integer $\lambda_j\in\{0,1,\ldots,k-\delta-i-1\}$.
The proof is complete. 
\end{proof}
  
\begin{remark}   
Theorem~\ref{thm:d4} along with Theorem~\ref{prop:partial_kalman_obse} indicates that when the conditions $Q_0=\hat Q_0\otimes pI_m$ and $\mu_0=\hat{ \mu}_0\otimes\mathds{1}_m$ are satisfied, there is no more space to further narrow the confidence interval of atomic time from the CKF algorithms. 
In practice, if the homogeneous atomic clocks utilized in the ensemble are located at the same laboratory, then the conditions for $Q_0$ and $\mu_0$ are usually realistic. 
This is because a reference signal such as UTC$[k]$ (Coordinated Universal Time) is usually received as an external source at the initial time $k=0$ for determining the initial guess $\hat{\bm x}[0]$,
where UTC$[k]$ can be regarded as the approximated ideal time with a tiny time deviation.
In such a case, the initial predicted error of the time deviation of the clocks can be regarded as the same with the variances being 0. 
However, if the clocks are located in different laboratories with independent receivers for UTC data, then the conditions may not hold since the received UTC$[0]$ may be diverse to the laboratories due to the synchronous requested timing in practice.
\end{remark}   


\section{Numerical Simulations} \label{sec:example}
This section provides an application example to demonstrate our results and compares the performance of the atomic time among the proposed method and the CKF algorithms with and without the existing covariance correction method by Greenhall \cite{greenhall2012review,godel2017kalman}.
In particular, we consider a third-order atomic clock ensemble with $m=5$ clocks.
The coefficients  of the model are set to  
 $ q_1^2=2.9394e-10$, $q_2^2=1.1785e-16$, and  $q_3^2= 4.5574e-35$.  
The sampling interval  is set to $\tau=1$.
The variances of measurement noises are set to $r_i^2=1e-12$, $i=1,2,3,4,5$.
In the following statements, we verify the results of Theorems~\ref{prop:partial_kalman_obse} and \ref{thm:d4}, and give a discussion for the robustness of the proposed structured Kalman filter concerning the initial prediction error covariance.

\subsubsection{Case 1} 
Letting the initial state and the initial guess be set to $\bm{x}[0]=\hat{\bm{x}}[0]=1e-28\mathds{1}_{15}$, and letting $\bm P_0=0.1I$, the confidence interval of the atomic time ${\rm TA}[k]$ 
 is shown in Fig.~\ref{TII_example1_confidence} to compare the proposed structured Kalman filter and the CKF algorithm with and without Greenhall's correction. 
 It follows from Theorem~\ref{prop:partial_kalman_obse} that the proposed method with $\hat{\bm P}_0=0.1(I_n\otimes\overline{V}\overline{V}^{\tr})$ and $\varGamma=0$ generates the same performance as the CKF algorithms if there is no computation error.
 which can be verified by the fact shown in Fig.~\ref{TII_example1_confidence}(a) that the black line representing the boundaries of the confidence interval of CKF algorithms in theory\footnote{The ideal confidence interval of CKF is generated by ${\rm TA}[k]=\frac{1}{m}C\bm{\epsilon}_{\rm \bar o}$ using \eqref{eq:weighted_error}.
The difference of CKF between real implementation (green) and ideal case (black) comes from the computation errors in numerical calculations which are diverging unboundedly along with the time.} coincides with the one of the proposed method with $\varGamma=0$.
It can be further seen from the figure that the proposed method can narrow the confidence interval from the CKF algorithms with  Greenhall's correction   \cite{greenhall2012review,godel2017kalman}, meaning that numerical stability is improved. 
Furthermore, the optimal $\varGamma$ is found as 0 in the optimization problem \eqref{eq:payoff}, which verifies Theorem~\ref{thm:d4}.

\subsubsection{Case 2}  
Let the initial state $\bm{x}[0]$ be set to a deterministic value around ${x}_{i}^j[0]\in (-6e-8,6e-8)$, whereas the initial predicted value is set to $\hat{\bm{x}}[0]=1e-28\mathds{1}_{15}$.
In this case, the expected value $\mu_0$ of initial prediction error does not satisfy the sufficient conditions in Theorem~\ref{thm:d4}.
The similar fact that the proposed method solves the numerical instability of CKF algorithms and the equivalence between the proposed method with $\hat{\bm P}_0=(I_n\otimes\overline{V})\bm P_0(I_n\otimes\overline{V})^{\tr}$ and the CKF algorithms can also be observed from the 98$\%$-confidence interval of ${\rm TA}[k]$  shown in Fig.~\ref{TII_example1_confidence}(b) with $\bm P_0=1e-2I$.
 Those observations also verify Theorem~\ref{prop:partial_kalman_obse}.
In terms of the optimal transformation matrix,
we use $\delta_1=1.0$, $\delta_2=5.4117$, $\hat{\bm P}_0=1e-4I$, and $T=1000$ to build the optimization problem \eqref{eq:payoff} and find the optimal transformation matrix $\varGamma$, which is found as a nonzero matrix. 
It can be seen from the yellow region which represents the 98$\%$-confidence interval of ${\rm TA}[k]$ under the optimal $\varGamma$ in Fig.~\ref{TII_example1_confidence}(b) is better maintained around zero than the case with $\varGamma=0$.



\begin{figure}
 \centering
\includegraphics[width=8.1cm]{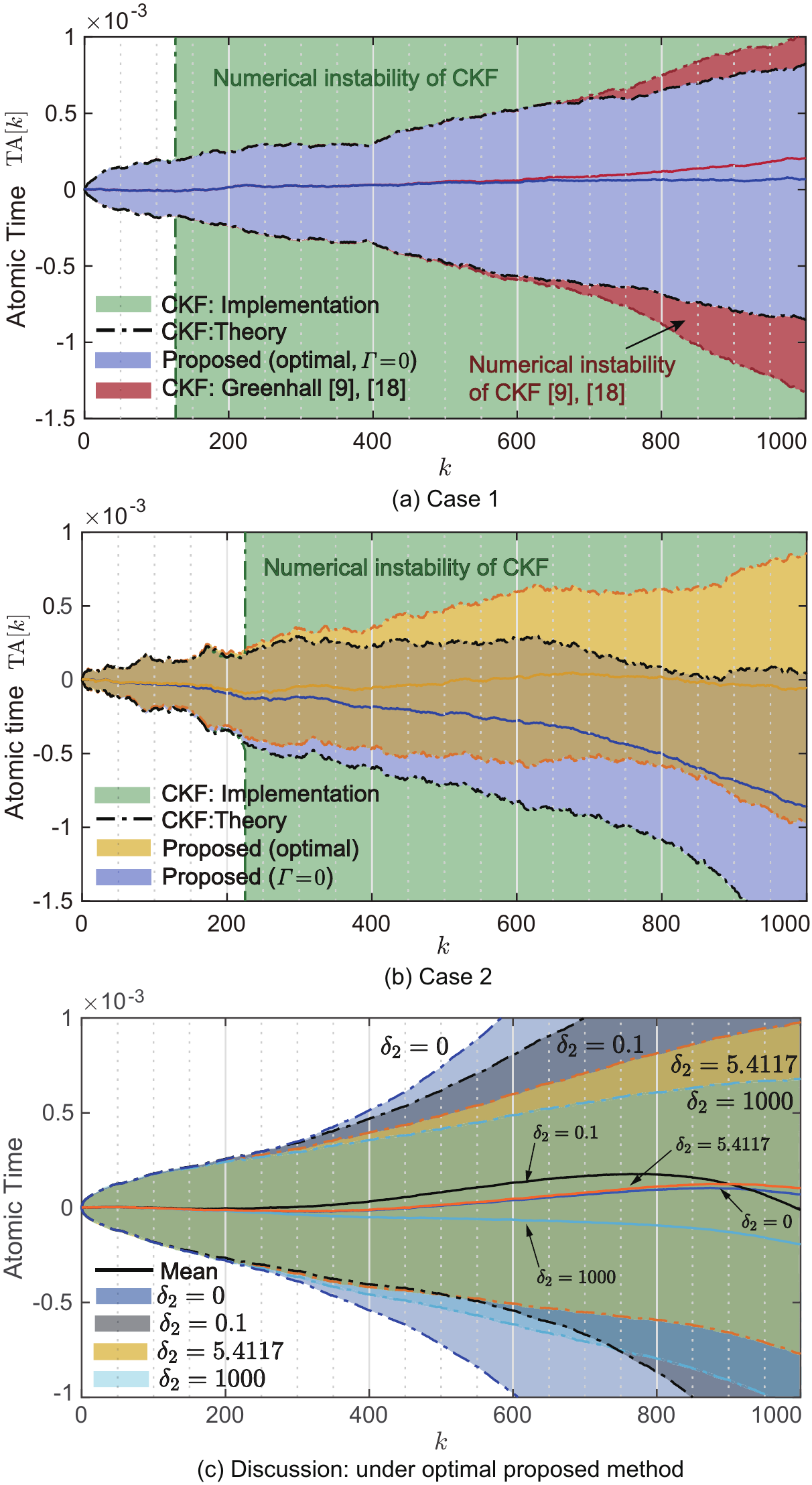} 
\caption{98$\%$-confidence interval of the \emph{atomic time} ${\rm TA}[k]$
under different algorithms. The black lines in (a) and (b) represent the ones under the CKF algorithms in ideal case. 
The smooth line in the middle of the confidence interval represents the expected value of ${\rm TA}[k]$ in 10 stochastic paths.
Both of (a) and (b) verify Theorem~\ref{prop:partial_kalman_obse},
whereas (a) further verifies Theorem~\ref{thm:d4}.
The atomic time can be better maintained around 0 under the proposed optimal method than the CKF algorithms without numerical instability issue.}
\label{TII_example1_confidence}
\end{figure}

\subsubsection{Discussion} 
The precise optimal solution of the optimization problem \eqref{eq:payoff} depends on the initial error covariance $\hat{\bm P}_0$ and the time $T$. 
However, ideally, for a given initial guess $\hat {\bm x}[0]$, the performance of both the CKF and our structured Kalman filter with a fixed $\varGamma$ does not significantly depend on the initial error covariance $\hat{\bm P}_0$ when $T$ is large enough.
This is because the atomic time ${\rm TA}[k+1]=\frac{1}{m}C\bm{\epsilon}_{\rm \bar o}[k+1]$ depends on the prediction error $\bm{\epsilon}_{\rm o}[k]$ and hence depends on the error covariance $\hat{\bm P}_{k-1}$ of the observable state.
Here, even though $\hat{\bm P}_{k}$ is ideally converging to the same steady-state value for any initial $\hat{\bm P}_{0}$, the value of $\hat{\bm P}_{k}$ may be diverging in the CKF algorithms for practical implementation due to the numerical instability and hence the atomic time  ${\rm TA}[k]$ may be significantly different from each other for different initial $\hat{\bm P}_{0}$ in the application.
In fact, this phenomenon does at appear in our proposed method.
Hence, the proposed optimal structured Kalman filter should show better robustness in terms of $\hat{\bm P}_0$ than the CKF algorithms.
To verify the robustness of the proposed method concerning the initial error covariance $\hat{\bm P}_0$, we use $\hat{\bm P}_0=0.01I,0.02I,0.04I$ in real implementation in the simulation. 
It can be seen from the trajectories of the atomic time ${\rm TA}[k]$ shown in Fig.~\ref{TII_example1_different_initial} that CKF algorithm with Greenhall's covariance correction method \cite{greenhall2012review,godel2017kalman} under $\bm P_0=0.01I,0.02I,0.04I$ shows worse robustness in Fig.~\ref{TII_example1_different_initial} since its generated atomic time ${\rm TA}[k]$ (solid lines in blue, green, and red) significantly depend on the initial $\bm P_0$. 
This fact can be also observed from the aspect of frequency stability of the predicted ensemble time deviations using the overlapping Allan deviations shown in Fig.~\ref{TII_Allan_1}.
It can be seen from the figure that the frequency stability of CKF with Greenhall's correction (orange lines) changes a lot when ${\bm P}_0=0.01I,0.02I,0.04I$ but it remains the same under our proposed structured Kalman filter (blue lines). 

The impact of the value of $\delta_2$ in the atomic time under the optimal proposed method with $\delta_1=1$ is illustrated in Fig.~\ref{TII_example1_confidence}(c). 
It can be seen from the figure that if $\delta_2$ is too small (e.g., $\delta_2=0$), the mean of the atomic time is very close to zero while the variance is large.
Slightly increasing $\delta_2$ can narrow the confidence interval without a big change in the mean value (see the case with $\delta_2=5.4117$).
However, if the $\delta_2$ is large enough, then increasing $\delta_2$ can make the mean of the atomic time worse (see the blue solid line representing the mean of the atomic time with $\delta_2=1000$).

\begin{figure}
 \centering
\includegraphics[width=7.6cm]{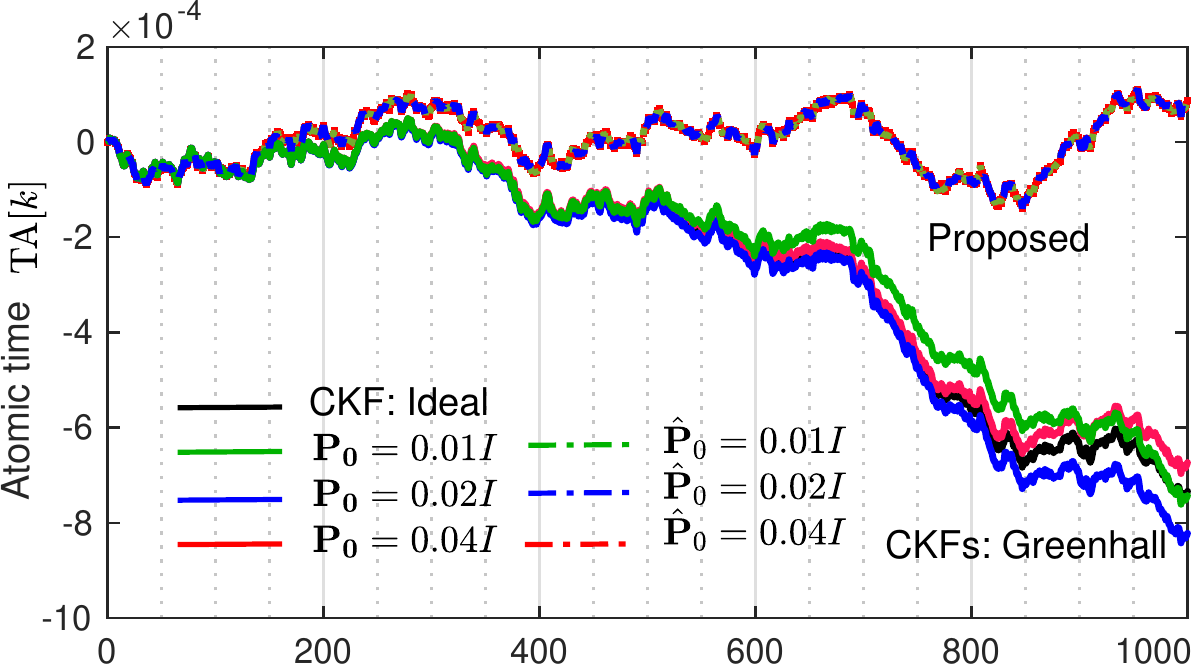} 
\caption{The \emph{atomic time} ${\rm TA}[k]$ under the proposed structured Kalman filter and CKF with Greenhall's correction correction method by Greenhall \cite{greenhall2012review,godel2017kalman}.
The black line represents the theoretical expression of CKF.}\label{TII_example1_different_initial}
\end{figure}

\begin{figure}
\centering
\includegraphics[width=8cm]{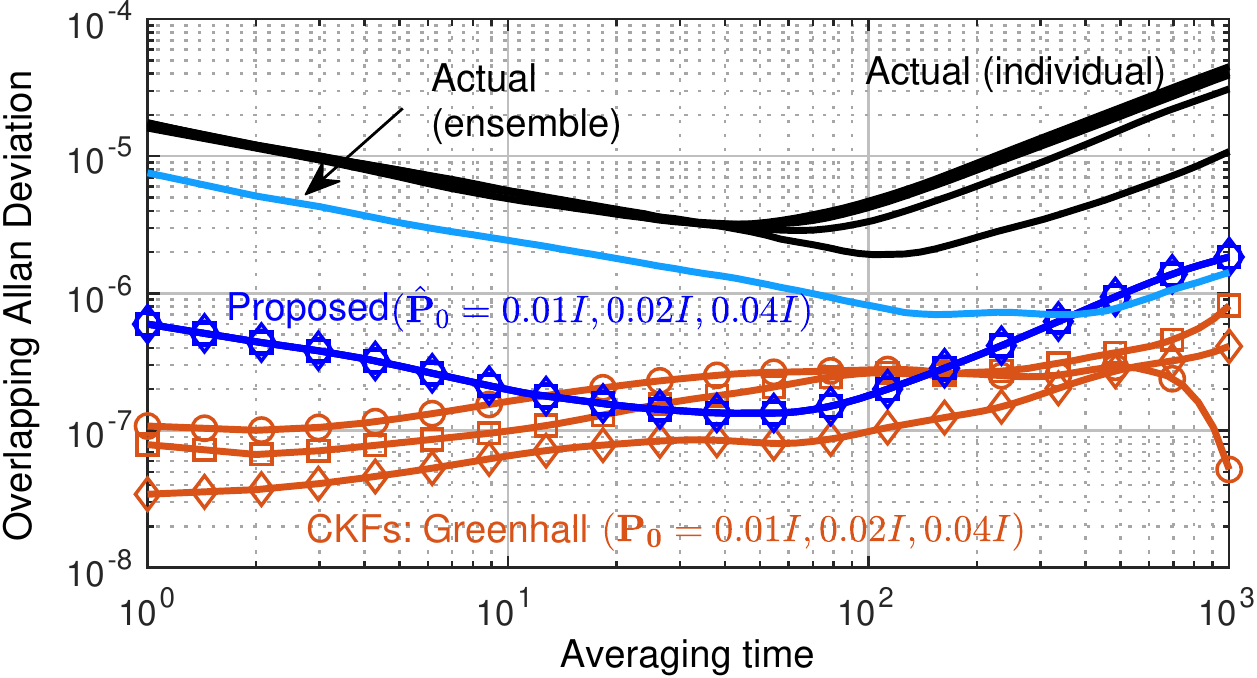} 
\caption{Overlapping Allan deviations of the predicted ensemble time deviation under the optimal PKF and CKF with Greenhall's covariance correction method \cite{greenhall2012review,godel2017kalman} (rhombus, square, and circle markers: $\bm P_0=0.01I,0,02I,0,04I$). 
The proposed optimal structured Kalman filter show better robustness in terms of $\hat{\bm P}_0$ than Greenhall's CKF algorithms.}\label{TII_Allan_1}
\end{figure}

\section{Conclusion} \label{sec:conc}
In this paper, we addressed a prediction problem of atomic clock ensembles for generating an accurate time scale and proposed a structured Kalman filter associated with the transformation matrix for observable Kalman canonical decomposition from CKF.
We presented the conditions where the proposed structured Kalman filter is reduced to the same as the CKF in the unobservable or observable state space. 
Due to undetectability of the clock model, we considered an objective function associated with not only the expected value of atomic time (prediction error in ensemble time deviation) but also its variance to find out the optimal transformation matrix in observable Kalman canonical decomposition.
We revealed that such an objective function is a convex function and showed some conditions under which CKF is nothing but the optimal algorithm minimizing the objective function.
A numerical example was presented in the paper to show the robustness of the proposed optimal structured Kalman filter in terms of the initial error covariance. 
It was found that the proposed method may further narrow the confidence interval from CKF as long as the transformation matrix can be properly picked.
\bibliographystyle{IEEEtran} 
 
\bibliography{Kalman}

\end{document}